\documentclass{article}
\usepackage[a4paper,margin=1in]{geometry}
\usepackage{amsmath,amssymb,amsthm,mathtools}
\usepackage{hyperref}
\usepackage{microtype}
\usepackage{amsfonts}
\usepackage{mathrsfs}
\usepackage{enumerate}

\theoremstyle{plain}
\newtheorem{theorem}{Theorem}[section]
\newtheorem{proposition}[theorem]{Proposition}
\newtheorem{corollary}[theorem]{Corollary}
\newtheorem{lemma}[theorem]{Lemma}
\theoremstyle{definition}
\newtheorem{definition}{Definition}[section]

\theoremstyle{remark}
\newtheorem{remark}{\textbf{Remark}}[section]

\theoremstyle{definition}

\usepackage{chngcntr}
\counterwithout{table}{section}

\numberwithin{equation}{section}

\newcommand{\ket}[1]{\lvert #1\rangle}

\title{Quantum-assisted rendezvous on hypercubes: a quantum Bernoulli noise approach}
\author{Ce Wang\thanks{email: cewangwhu@163.com; orcid:0000-0002-2916-8478} \\
Shanghai Institute for Mathematics and Interdisciplinary Sciences (SIMIS)\\
Shanghai 200433, China}
\date{}

\begin{document}
\maketitle
\begin{abstract}
We study quantum-assisted balanced rendezvous on the
$d$-dimensional hypercube in a one-step two-agent quantum-walk
framework, where the two agents move the same Hamming distance and meet
at a common midpoint.
The key structure is a projective $t$-sphere resolution, which governs
exact finite-dimensional quantum attainment. Its scalar analogue
recovers the corresponding classical characterization. We characterize
finite-dimensional quantum attainment of the no-signalling benchmark,
construct quantum strategies from projective sphere resolutions, and
establish a rigidity theorem for optimal finite-dimensional strategies.
A marked difference appears between nearest-neighbor and
higher-cardinality updates: quantum advantage can occur in the former,
whereas in the latter the geometry of the relevant Hamming spheres
rules out exact finite-dimensional attainment of the no-signalling
value. Together, these results connect quantum-walk dynamics,
Hamming-sphere geometry, and quantum rendezvous within the
quantum Bernoulli noises (QBN) framework.
\end{abstract}

\vskip 2mm

\noindent\textbf{Keywords.}\ \ quantum Bernoulli noises; balanced rendezvous;
nonlocal games; hypercube; projective Hamming-sphere resolutions;
Hadamard matrices
\vskip 2mm

\noindent\textbf{Mathematics Subject Classification (2020).}\ 81P45, 81Q35, 05C57, 05B20

\section{Introduction}
\label{sec:introduction}
Rendezvous problems ask whether spatially separated agents, acting
without communication after receiving their inputs, can coordinate
their local moves so as to reach the same state. Classical rendezvous
search has a long history in search theory and graph-based coordination
\cite{Alpern1995}. Recent quantum versions have shown that shared
entanglement can improve rendezvous performance on specific graphs,
through both explicit protocols and experimental implementations, and
the setting has recently been extended to graph ensembles with
uncertain topology
\cite{Mironowicz2023,ViolaMironowicz2024,Tucker2024,
TuckerEtAl2026GraphEnsembles}. More broadly, nonlocal games provide a
natural framework for distributed quantum coordination
\cite{DingJiang2024,DingXu2026}, while families of projections indexed
by graph structures arise in the characterization of perfect quantum
strategies
\cite{MancinskaRoberson2016,CuiMancinskaNezhadiRoberson2025}.
These developments provide bounds, protocols, and structural
descriptions for a range of graph-based coordination problems. For
balanced rendezvous on the hypercube, however, it remains natural to
ask whether exact attainment of the no-signalling benchmark can be
characterized by an intrinsic combinatorial or operator structure.

We consider a balanced rendezvous problem in which Alice and Bob
receive two hypercube vertices at Hamming distance $2t$ and
independently choose updates of cardinality $t$, with the aim of
reaching the same vertex. We study this problem using quantum Bernoulli
noises (QBN). In the QBN framework, the operators
$\partial_k+\partial_k^*$ realize individual coordinate flips
\cite{WangChaiLu2010,WangYe2016,WangHypercube2022}, and hence provide a
direct operator realization of the allowed local updates. Rendezvous
occurs exactly when the two updates form a disjoint partition of the
relative coordinate $\eta=\sigma\triangle\tau$. For $t=1$, this is the
distance-two balanced version of nearest-neighbor rendezvous, whereas
$t>1$ corresponds to a single fixed-cardinality update rather than a
multi-step walk.

The QBN realization gives a one-step two-agent quantum-walk
description of the protocol, with the two walkers correlated through
their coin degrees of freedom and success determined by coincidence of
their final positions. Since rendezvous depends only on the common
final vertex, the associated conditions on the coin projections lead
to projective $t$-sphere resolutions, while scalar resolutions describe
the corresponding classical structure. The no-signalling benchmark can
be computed explicitly, and for $d>2t$ its finite-dimensional quantum
attainment is equivalent to the existence of a projective
$t$-sphere resolution. Equivalently, such attainment occurs precisely
when a rendezvous-compatible $t$-coin field exists. The scalar analogue
gives the corresponding criterion for classical saturation.

The behavior is different for $t=1$ and $t\ge2$. For $t=1$,
classical saturation occurs exactly when $d=2^m$, recovering the
classical total-perfect-code criterion for the hypercube
\cite{KwonLee2014}. On the quantum side, every
real Hadamard matrix of order $d$ yields a rank-one projective
$1$-sphere resolution and hence a finite-dimensional strategy attaining
$\omega_{\mathrm{NS}}(d)=2/d$. Thus, whenever a Hadamard matrix exists
and $d$ is not a power of two, the balanced rendezvous problem exhibits
a strict quantum advantage. For $t\ge2$ and $d>2t$, translation
symmetrization together with a rank argument rules out all
finite-dimensional projective $t$-sphere resolutions, and therefore
rules out exact finite-dimensional attainment of the no-signalling
value. This shows that the geometry of the possible meeting vertices
controls exact attainment: projective sphere resolutions permit a
quantum separation in the nearest-neighbor case, while the
higher-cardinality geometry obstructs the same mechanism for
$t\ge2$. Whether
$\omega_{\mathrm Q}(d,t)<\omega_{\mathrm{NS}}(d,t)$ and whether a
strict classical--quantum separation persists in this regime remain
open.

\section{Mathematical formulation of hypercube rendezvous}
\label{sec:problem}
We formulate the rendezvous problem using the power-set realization
of the \(d\)-dimensional hypercube. Let
\[
N_{d-1}:=\{0,\ldots,d-1\},
\qquad
\Gamma_{d-1}:=2^{N_{d-1}},
\]
so that each \(\sigma\in\Gamma_{d-1}\) represents a vertex of the
\(d\)-dimensional hypercube. For \(0\le r\le d\), write
\[
\Omega_r
:=
\{\sigma\in\Gamma_{d-1}:|\sigma|=r\}.
\]

We consider a balanced rendezvous problem in which the two agents start
at Hamming distance \(2t\), and each agent changes exactly \(t\)
coordinates. Thus both agents move the same Hamming distance toward a
common final vertex.

Fix
\(
1\le t\le\frac d2.
\)
The input pairs are
\begin{equation}
\mathcal I_{d,t}
:=
\left\{
(\sigma,\tau)\in\Gamma_{d-1}^2:
|\sigma\triangle\tau|=2t
\right\}.
\label{eq:balanced-input-set}
\end{equation}
For \((\sigma,\tau)\in\mathcal I_{d,t}\), define the relative set
\begin{equation}
\eta
:=
\sigma\triangle\tau
\in\Omega_{2t}.
\label{eq:relative-coordinate}
\end{equation}

Alice receives \(\sigma\), Bob receives \(\tau\), and no communication
is allowed after the inputs are given. Each player chooses a
\(t\)-element subset of \(N_{d-1}\). We denote these subsets by
\[
A,B\in\Omega_t.
\]
Their vertices are updated by symmetric difference:
\begin{equation}
\sigma'
:=
\sigma\triangle A,
\qquad
\tau'
:=
\tau\triangle B.
\label{eq:local-updates}
\end{equation}
The task succeeds when
\begin{equation}
\sigma'=\tau'.
\label{eq:success_condition}
\end{equation}

For $t=1$, this is the distance-two balanced version of
nearest-neighbor rendezvous on the hypercube. For \(t>1\), each player performs
a single fixed-cardinality update of size \(t\), rather than a
sequence of \(t\) nearest-neighbor moves.

Using \(\tau=\sigma\triangle\eta\),
equation \eqref{eq:success_condition} is equivalent to
\[
A=\eta\triangle B.
\]
Since
\[
|A|=|B|=t,
\qquad
|\eta|=2t,
\]
this is equivalent to
\begin{equation}
A\cap B=\varnothing,
\qquad
A\cup B=\eta.
\label{eq:balanced-rendezvous-condition}
\end{equation}
Thus successful rendezvous means that Alice and Bob partition the
\(2t\)-element set \(\eta\) into two disjoint \(t\)-subsets.
Define the success indicator by
\begin{equation}
\chi_{d,t}(\eta;A,B)
:=
\mathbf 1_{\{A\cap B=\varnothing,\ A\cup B=\eta\}}.
\label{eq:success-indicator}
\end{equation}
A strategy is described by its conditional joint distribution
\[
\mathbb P(A,B\mid\sigma,\tau),
\qquad
A,B\in\Omega_t,
\]
where \(\mathbb P(A,B\mid\sigma,\tau)\) is the probability that Alice chooses
\(A\) and Bob chooses \(B\), given the input vertices
\(\sigma\) and \(\tau\).

We assume that the  input pairs in
\(\mathcal I_{d,t}\) are chosen uniformly. Every such pair can be
written uniquely as
\[
(\sigma,\tau)
=
(\sigma,\sigma\triangle\eta),
\qquad
\sigma\in\Gamma_{d-1},
\quad
\eta\in\Omega_{2t}.
\]
Hence the success probability of a strategy \(\mathbb P\) is
\begin{equation}
\omega_{d,t}(\mathbb P)
:=
\frac{1}{|\mathcal I_{d,t}|}
\sum_{\sigma\in\Gamma_{d-1}}
\sum_{\eta\in\Omega_{2t}}
\sum_{A,B\in\Omega_t}
\chi_{d,t}(\eta;A,B)
\mathbb P(A,B\mid\sigma,\sigma\triangle\eta).
\label{eq:rendezvous-functional}
\end{equation}
Since
\[
|\mathcal I_{d,t}|
=
2^d\binom d{2t},
\]
the normalization in \eqref{eq:rendezvous-functional} is the
uniform average over all  input pairs.

Let
\[
\mathcal C_{d,t},
\qquad
\mathcal Q_{d,t},
\qquad
\mathcal{NS}_{d,t}
\]
denote the sets of conditional joint distributions induced by
classical, finite-dimensional quantum, and no-signalling strategies,
respectively. 

Following the standard convention for nonlocal games
\cite{CleveHoyerTonerWatrous2004,BrunnerEtAl2014,Slofstra2019},
we define
\begin{equation}
\begin{aligned}
\omega_{\mathrm C}(d,t)
&:=
\max_{\mathbb P\in\mathcal C_{d,t}}
\omega_{d,t}(\mathbb P),\\
\omega_{\mathrm Q}(d,t)
&:=
\sup_{\mathbb P\in\mathcal Q_{d,t}}
\omega_{d,t}(\mathbb P),\\
\omega_{\mathrm{NS}}(d,t)
&:=
\max_{\mathbb P\in\mathcal{NS}_{d,t}}
\omega_{d,t}(\mathbb P).
\end{aligned}
\label{eq:classical-quantum-ns-values}
\end{equation}
Since
\begin{equation}
\mathcal C_{d,t}
\subseteq
\mathcal Q_{d,t}
\subseteq
\mathcal{NS}_{d,t},
\label{eq:correlation-set-inclusions}
\end{equation}
we have
\begin{equation}
\omega_{\mathrm C}(d,t)
\le
\omega_{\mathrm Q}(d,t)
\le
\omega_{\mathrm{NS}}(d,t).
\label{eq:value-hierarchy}
\end{equation}
Here the classical and no-signalling optima are maxima because the
corresponding correlation sets are compact polytopes, whereas the
quantum value is written as a supremum since the finite-dimensional
quantum correlation set need not be closed.

The central question is whether quantum strategies can outperform
classical ones. Moreover, we ask when a finite-dimensional quantum strategy can
attain the no-signalling benchmark, and what operator structure is
required for such exact attainment.

%
\section{QBN formulation of the rendezvous conditions}
\label{sec:qbn-geometry}

We now pass from the combinatorial description of the hypercube to its
QBN realization introduced in \cite{WangHypercube2022}. Let
\(\mathfrak h\) be the Hilbert space of square-integrable Bernoulli
functionals with canonical orthonormal basis
\[
\{Z_\sigma:\sigma\subset\mathbb N,\ |\sigma|<\infty\}.
\]
For \(k\in\mathbb N\), the QBN annihilation and creation operators act
on the Bernoulli basis by
\[
\partial_kZ_\sigma
=
\mathbf1_\sigma(k)Z_{\sigma\setminus\{k\}},
\qquad
\partial_k^*Z_\sigma
=
\bigl(1-\mathbf1_\sigma(k)\bigr)Z_{\sigma\cup\{k\}}.
\]
Restricting to the first \(d\) coordinates, set
\[
\mathfrak h_{d-1}
:=
\operatorname{span}
\{Z_\sigma:\sigma\in\Gamma_{d-1}\}.
\]
Thus every hypercube vertex
\(\sigma\in\Gamma_{d-1}\) is represented by the Bernoulli basis vector
\(Z_\sigma\in\mathfrak h_{d-1}\).

For \(k\in N_{d-1}\), define the QBN shift
\[
\Xi_k
:=
\partial_k+\partial_k^*.
\]
It satisfies
\[
\Xi_kZ_\sigma
=
Z_{\sigma\triangle\{k\}},
\qquad
\Xi_k^*=\Xi_k,
\qquad
\Xi_k^2=I,
\qquad
\Xi_j\Xi_k=\Xi_k\Xi_j.
\]
Hence, for \(A\subseteq N_{d-1}\), the fixed-cardinality update
\(\sigma\mapsto\sigma\triangle A\) is realized by
\begin{equation}
\Xi_A
:=
\prod_{k\in A}\Xi_k,
\qquad
\Xi_AZ_\sigma
=
Z_{\sigma\triangle A}.
\label{eq:qbn-coordinate-update}
\end{equation}
Thus the local update rule
\eqref{eq:local-updates} has an exact QBN operator realization.

The rendezvous condition also admits a direct QBN representation.
Indeed,
\[
\Xi_A\Xi_BZ_\eta
=
Z_{\eta\triangle A\triangle B},
\]
and therefore
\begin{equation}
\chi_{d,t}(\eta;A,B)
=
\left|
\left\langle
Z_\varnothing,
\Xi_A\Xi_BZ_\eta
\right\rangle
\right|^2.
\label{eq:balanced-winning-kernel}
\end{equation}
Thus the successful branches are exactly those for which the combined
QBN shifts return the relative-coordinate state \(Z_\eta\) to the
vacuum vector \(Z_\varnothing\).

For later use, let
\[
V_t
:=
\operatorname{span}
\{Z_A:A\in\Omega_t\}
\subset\mathfrak h_{d-1}.
\]
The disjointness condition \(A\cap B=\varnothing\), which is part of
\eqref{eq:balanced-rendezvous-condition}, induces a useful connectivity
property among the Bernoulli basis vectors of \(V_t\).

\begin{proposition}
\label{prop:even-disjointness-chain}
Assume \(d>2t\). Let \(Z_A,Z_{A'}\in V_t\), and set
\[
m
:=
|A\setminus A'|
=
|A'\setminus A|.
\]
Then there exist
\[
A_0,B_0,A_1,B_1,\ldots,B_{m-1},A_m\in\Omega_t
\]
with
\[
A_0=A,
\qquad
A_m=A',
\]
such that
\begin{equation}
A_i\cap B_i=\varnothing,
\qquad
B_i\cap A_{i+1}=\varnothing,
\qquad
0\le i<m.
\label{eq:even-disjointness-chain}
\end{equation}
Consequently, \(Z_A\) and \(Z_{A'}\) can be connected by a chain of
\(2m\) disjointness transitions in \(V_t\).
\end{proposition}

\begin{proof}
Write
\[
A\setminus A'
=
\{a_1,\ldots,a_m\},
\qquad
A'\setminus A
=
\{a'_1,\ldots,a'_m\},
\]
and define
\[
A_i
:=
\bigl(
A\setminus\{a_1,\ldots,a_i\}
\bigr)
\cup
\{a'_1,\ldots,a'_i\},
\qquad
0\le i\le m.
\]
Then \(A_i\in\Omega_t\),
\(A_0=A\), \(A_m=A'\), and
\[
|A_i\triangle A_{i+1}|=2.
\]
Hence
\[
|A_i\cup A_{i+1}|=t+1.
\]
Since \(d>2t\), we have \(d\ge2t+1\), and therefore
\[
\left|
N_{d-1}\setminus(A_i\cup A_{i+1})
\right|
=
d-t-1
\ge t.
\]
We may thus choose
\[
B_i
\in\Omega_t,
\qquad
B_i
\subseteq
N_{d-1}\setminus(A_i\cup A_{i+1}).
\]
It follows that
\[
A_i\cap B_i=\varnothing,
\qquad
B_i\cap A_{i+1}=\varnothing.
\]
Thus
\[
Z_{A_i}
\longrightarrow
Z_{B_i}
\longrightarrow
Z_{A_{i+1}}
\]
is a pair of consecutive disjointness transitions. Repeating this
construction for \(i=0,\ldots,m-1\) gives
\[
Z_{A_0}
\longrightarrow
Z_{B_0}
\longrightarrow
Z_{A_1}
\longrightarrow
\cdots
\longrightarrow
Z_{B_{m-1}}
\longrightarrow
Z_{A_m}.
\]
The resulting chain connects \(Z_A\) to \(Z_{A'}\) through exactly
\(2m\) disjointness transitions.
\end{proof}

Proposition~\ref{prop:even-disjointness-chain} provides the
connectivity property needed in the rigidity argument of
Theorem~\ref{thm:quantum-saturation-rigidity}. Together with
\eqref{eq:qbn-coordinate-update} and
\eqref{eq:balanced-winning-kernel}, it gives the QBN operator
description required for the rendezvous problem. We next ask how
these local QBN updates can be organized by their common final vertex.
\section{Updates and projective sphere resolutions}
\label{sec:coin-systems}
For $A\in\Omega_t$, the update
\[
\sigma\longmapsto\sigma\triangle A
\]
is realized by the QBN shift $\Xi_A$ and moves $\sigma$ to a vertex at
Hamming distance $t$. In particular, for $t=1$, each allowed local
update is a single nearest-neighbor step of the QBN quantum walk, with
the singleton update $\{k\}$ implemented by
$\Xi_k$ \cite{WangHypercube2022}.

To encode these updates in the quantum model, let
\(\mathcal K\) be a finite-dimensional Hilbert space and associate a
coin operator \(C_A\in\mathcal B(\mathcal K)\) with each
\(A\in\Omega_t\). Set
\[
\mathcal C
=
\{C_A:A\in\Omega_t\}.
\]
We call \(\mathcal C\) a coin operator system for \(t\)-element
updates if
\[
C_A^*C_B
=
C_AC_B^*
=
0,
\qquad
A\ne B,
\]
and
\[
U_{\mathcal C}
:=
\sum_{A\in\Omega_t}C_A
\]
is unitary.

For
\[
P_A:=C_A^*C_A,
\]
these conditions are equivalent to
\begin{equation}
P_AP_B=0
\quad(A\ne B),
\qquad
\sum_{A\in\Omega_t}P_A=I,
\qquad
C_A=U_{\mathcal C}P_A.
\label{eq:coin-polar-form}
\end{equation}

\begin{proposition}
\label{prop:t-coordinate-qbn-operator}
Let
\[
\mathcal C=\{C_A:A\in\Omega_t\}
\]
be a coin operator system for \(t\)-element updates. Then
\begin{equation}
W_{\mathcal C}^{[t]}
:=
\sum_{A\in\Omega_t}
\Xi_A\otimes C_A
\label{eq:t-coordinate-qbn-operator}
\end{equation}
is unitary on \(\mathfrak h_{d-1}\otimes\mathcal K\), and
\begin{equation}
W_{\mathcal C}^{[t]}
(Z_\sigma\otimes\xi)
=
\sum_{A\in\Omega_t}
Z_{\sigma\triangle A}\otimes C_A\xi,
\qquad
\xi\in\mathcal K.
\label{eq:t-coordinate-qbn-operator-action}
\end{equation}
\end{proposition}

\begin{proof}
Since every \(\Xi_A\) is a self-adjoint unitary,
\begin{align*}
(W_{\mathcal C}^{[t]})^*
W_{\mathcal C}^{[t]}
&=
\sum_{A,B\in\Omega_t}
\Xi_A\Xi_B\otimes C_A^*C_B\\
&=
I\otimes
\sum_{A\in\Omega_t}C_A^*C_A\\
&=
I.
\end{align*}
Similarly,
\[
W_{\mathcal C}^{[t]}
(W_{\mathcal C}^{[t]})^*
=
I.
\]
Equation \eqref{eq:t-coordinate-qbn-operator-action} follows from
\eqref{eq:qbn-coordinate-update}.
\end{proof}

\begin{remark}
The superscript \([t]\) specifies the update distance, not the number
of time steps. In particular,
\[
W_{\mathcal C}^{[t]}
\ne
\bigl(W_{\mathcal C}^{[1]}\bigr)^t.
\]
\end{remark}

For fixed \(\sigma\in\Gamma_{d-1}\), each \(A\in\Omega_t\) determines
the final vertex
\[
\zeta=\sigma\triangle A.
\]
Since successful rendezvous requires the two players to have the
same \(\zeta\), we associate a projection \(Q_\zeta\) with each
\(\zeta\in\Gamma_{d-1}\). Since
\[
\{\sigma\triangle A:A\in\Omega_t\}
=
\{\zeta\in\Gamma_{d-1}:|\sigma\triangle\zeta|=t\},
\]
this leads to the following resolution condition.

\begin{definition}
\label{def:projective-t-sphere}
A projective \(t\)-sphere resolution on a finite-dimensional Hilbert
space \(\mathcal K\) is a family of projections
\[
Q_\zeta
=
Q_\zeta^*
=
Q_\zeta^2
\in\mathcal B(\mathcal K),
\qquad
\zeta\in\Gamma_{d-1},
\]
such that
\begin{equation}
\sum_{A\in\Omega_t}
Q_{\sigma\triangle A}
=
I_{\mathcal K},
\qquad
\sigma\in\Gamma_{d-1}.
\label{eq:projective-t-sphere}
\end{equation}
\end{definition}

We next require the domain projection of the coin operator associated
with \(A\) at input \(\sigma\) to depend only on
\(\sigma\triangle A\).

\begin{definition}
\label{def:compatible-coin-field}
A rendezvous-compatible \(t\)-coin field is a family
\[
\mathcal C^{(\sigma)}
=
\{C_A^{(\sigma)}:A\in\Omega_t\},
\qquad
\sigma\in\Gamma_{d-1},
\]
such that each \(\mathcal C^{(\sigma)}\) is a coin operator system for
\(t\)-element updates and there exists a family of projections
\[
\{Q_\zeta:\zeta\in\Gamma_{d-1}\}
\]
satisfying
\begin{equation}
C_A^{(\sigma)*}C_A^{(\sigma)}
=
Q_{\sigma\triangle A},
\qquad
\sigma\in\Gamma_{d-1},
\quad
A\in\Omega_t.
\label{eq:meeting-point-compatibility}
\end{equation}
\end{definition}

\begin{theorem}
\label{thm:coin-resolution-equivalence}
A projective \(t\)-sphere resolution on \(\mathcal K\) exists if and
only if a rendezvous-compatible \(t\)-coin field exists.

More precisely, if
\(\{Q_\zeta\}_{\zeta\in\Gamma_{d-1}}\)
is a projective \(t\)-sphere resolution and
\(\{U_\sigma\}_{\sigma\in\Gamma_{d-1}}\)
is any family of unitaries on \(\mathcal K\), then
\begin{equation}
C_A^{(\sigma)}
:=
U_\sigma Q_{\sigma\triangle A}
\label{eq:coin-field-from-resolution}
\end{equation}
defines a rendezvous-compatible \(t\)-coin field.
\end{theorem}

\begin{proof}
Suppose first that
\(\{Q_\zeta\}_{\zeta\in\Gamma_{d-1}}\)
satisfies \eqref{eq:projective-t-sphere}. For fixed \(\sigma\), the
projections
\[
\{Q_{\sigma\triangle A}:A\in\Omega_t\}
\]
form an orthogonal resolution of \(\mathcal K\). Hence, for
\(A\ne B\),
\[
C_A^{(\sigma)*}C_B^{(\sigma)}
=
Q_{\sigma\triangle A}
Q_{\sigma\triangle B}
=
0,
\]
and
\[
C_A^{(\sigma)}
C_B^{(\sigma)*}
=
U_\sigma
Q_{\sigma\triangle A}
Q_{\sigma\triangle B}
U_\sigma^*
=
0.
\]
Moreover,
\[
\sum_{A\in\Omega_t}
C_A^{(\sigma)}
=
U_\sigma
\sum_{A\in\Omega_t}
Q_{\sigma\triangle A}
=
U_\sigma.
\]
Thus each \(\mathcal C^{(\sigma)}\) is a coin operator system, and
\eqref{eq:meeting-point-compatibility} holds.

Conversely, let a rendezvous-compatible \(t\)-coin field be given.
Then, for every \(\sigma\in\Gamma_{d-1}\),
\[
I_{\mathcal K}
=
\sum_{A\in\Omega_t}
C_A^{(\sigma)*}C_A^{(\sigma)}
=
\sum_{A\in\Omega_t}
Q_{\sigma\triangle A},
\]
which is \eqref{eq:projective-t-sphere}.
\end{proof}

For each \(\sigma\in\Gamma_{d-1}\), the compatible coin system
\(\mathcal C^{(\sigma)}\) defines
\begin{equation}
W_{\mathcal C^{(\sigma)}}^{[t]}
:=
\sum_{A\in\Omega_t}
\Xi_A\otimes C_A^{(\sigma)}.
\label{eq:input-conditioned-qbn-update}
\end{equation}
These operators give the local quantum updates for the two players.
Accordingly, the rendezvous problem may be viewed as a one-step
two-agent QBN quantum walk on the hypercube, with rendezvous
corresponding to coincidence of the two final positions. In the next
section, we combine these local updates with a shared entangled state
and show how the resulting position measurements induce a quantum
strategy for the rendezvous game.

\section{Saturation criteria and rigidity}
\label{sec:saturation-rigidity}
We first determine the information-theoretic benchmark independently
of the QBN construction. The no-signalling relaxation of rendezvous
games was considered in
\cite{Mironowicz2023,Tucker2024}. For the balanced
fixed-cardinality model considered here, the corresponding value has
the following explicit form.

We say that a
classical or quantum strategy \emph{saturates the no-signalling
benchmark} if its success probability equals
\(\omega_{\mathrm{NS}}(d,t)\).
\begin{lemma}
\label{lem:NS-general}
For \(1\le t\le d/2\),
\begin{equation}
\omega_{\mathrm{NS}}(d,t)
=
\frac{\binom{d-t}{t}}{\binom d{2t}}
=
\frac{\binom{2t}{t}}{\binom dt}.
\label{eq:NS-general-value}
\end{equation}
\end{lemma}

The proof is given in
Appendix~\ref{app:no-signalling}.

Classical saturation admits a corresponding characterization in terms
of scalar sphere resolutions.

\begin{proposition}
\label{prop:classical-general}
Assume \(d>2t\). Then
\begin{equation}
\omega_{\mathrm C}(d,t)
=
\omega_{\mathrm{NS}}(d,t)
\label{eq:classical-general-saturation}
\end{equation}
if and only if there exists a function
\[
q:\Gamma_{d-1}\longrightarrow\{0,1\}
\]
such that
\begin{equation}
\sum_{A\in\Omega_t}
q(\sigma\triangle A)
=
1,
\qquad
\sigma\in\Gamma_{d-1}.
\label{eq:scalar-t-sphere-resolution}
\end{equation}
That is, classical saturation occurs if and only if a scalar
\(t\)-sphere resolution exists.
\end{proposition}

The proof is given in
Appendix~\ref{app:classical-saturation}.

We now return to the QBN realization and show how a projective
\(t\)-sphere resolution produces a quantum strategy attaining the
no-signalling benchmark. Let \(D:=\dim\mathcal K\), and fix an
orthonormal basis \(\{e_\mu\}_{\mu=1}^D\) of \(\mathcal K\). Let
\[
\mathfrak C
:=
\{\mathcal C^{(\sigma)}:
\sigma\in\Gamma_{d-1}\}
\]
be a compatible \(t\)-coin field associated with a projective
\(t\)-sphere resolution.

For an input pair \((\sigma,\tau)\), define the entangled initial state
on
\[
(\mathfrak h_{d-1}\otimes\mathcal K)
\otimes
(\mathfrak h_{d-1}\otimes\mathcal K)
\]
by
\begin{equation}
\ket{\Psi_{\sigma,\tau}^{\mathrm{in}}}
:=
\frac{1}{\sqrt D}
\sum_{\mu=1}^D
\bigl(Z_\sigma\otimes e_\mu\bigr)
\otimes
\bigl(Z_\tau\otimes e_\mu\bigr).
\label{eq:entangled-initial-state}
\end{equation}
Here the position encode the local inputs, while the
entanglement is entirely contained in the coin factors and is
independent of \(\sigma\) and \(\tau\).

Alice and Bob apply the local QBN updates
\(W_{\mathcal C^{(\sigma)}}^{[t]}\) and
\(\overline{W_{\mathcal C^{(\tau)}}^{[t]}}\),  as defined
in \eqref{eq:input-conditioned-qbn-update}. Here the bar denotes complex conjugation with respect to the fixed
coin basis \(\{e_\mu\}_{\mu=1}^D\) and the Bernoulli basis on the
position space. The resulting joint state is
\begin{equation}
\ket{\Psi_{\sigma,\tau}^{\mathrm{out}}}
:=
\left(
W_{\mathcal C^{(\sigma)}}^{[t]}
\otimes
\overline{W_{\mathcal C^{(\tau)}}^{[t]}}
\right)
\ket{\Psi_{\sigma,\tau}^{\mathrm{in}}}.
\label{eq:qbn-updated-joint-state}
\end{equation}
Thus \eqref{eq:qbn-updated-joint-state} describes a one-step
two-agent QBN quantum walk on the hypercube with Hamming step length
\(t\), with the two walkers correlated through their entangled coin
degrees of freedom.

For fixed \(\sigma\), the map
\[
A\longmapsto\sigma\triangle A,
\qquad
A\in\Omega_t,
\]
is injective, and hence measuring the final position
\(\sigma\triangle A\) identifies the output \(A\) uniquely. The same
holds for Bob. For \(\zeta\in\Gamma_{d-1}\), let
\(\Pi_\zeta:=|Z_\zeta\rangle\langle Z_\zeta|\). Therefore the conditional output probability of the
quantum strategy is exactly the Born probability induced by the
updated QBN state:
\begin{equation}
\begin{aligned}
\mathbb P_{\mathfrak C}(A,B\mid\sigma,\tau)
={}&
\Bigl\|
\bigl[
(\Pi_{\sigma\triangle A}\otimes I_{\mathcal K})
\otimes
(\Pi_{\tau\triangle B}\otimes I_{\mathcal K})
\bigr]
\ket{\Psi_{\sigma,\tau}^{\mathrm{out}}}
\Bigr\|^2.
\end{aligned}
\label{eq:qbn-born-correlation}
\end{equation}

\begin{proposition}
\label{prop:qbn-update-correlation}
Let
\(\{Q_\zeta\}_{\zeta\in\Gamma_{d-1}}\)
be a projective \(t\)-sphere resolution and let \(\mathfrak C\) be a
compatible \(t\)-coin field. Then
\begin{equation}
\mathbb P_{\mathfrak C}(A,B\mid\sigma,\tau)
=
\frac{1}{D}
\operatorname{Tr}\!\left(
Q_{\sigma\triangle A}
Q_{\tau\triangle B}
\right).
\label{eq:qbn-update-correlation}
\end{equation}
In particular, if
\(\sigma\triangle A=\tau\triangle B=\zeta\), then
\begin{equation}
\mathbb P_{\mathfrak C}(A,B\mid\sigma,\tau)
=
\frac{\operatorname{rank}Q_\zeta}{D}.
\label{eq:meeting-branch-probability}
\end{equation}
\end{proposition}

\begin{proof}
Since
\(\Xi_A Z_\sigma=Z_{\sigma\triangle A}\) and
\(\Xi_B Z_\tau=Z_{\tau\triangle B}\), the position projections in
\eqref{eq:qbn-born-correlation} select exactly the branch
\((A,B)\). Therefore
\[
\mathbb P_{\mathfrak C}(A,B\mid\sigma,\tau)
=
\frac{1}{D}
\left\|
\sum_{\mu=1}^D
C_A^{(\sigma)}e_\mu
\otimes
\overline{C_B^{(\tau)}e_\mu}
\right\|^2.
\]
Expanding the norm in the orthonormal basis
\(\{e_\mu\}_{\mu=1}^D\) gives
\[
\mathbb P_{\mathfrak C}(A,B\mid\sigma,\tau)
=
\frac{1}{D}
\operatorname{Tr}\!\left(
C_A^{(\sigma)*}C_A^{(\sigma)}
C_B^{(\tau)*}C_B^{(\tau)}
\right).
\]
By the compatibility condition,
\[
C_A^{(\sigma)*}C_A^{(\sigma)}
=
Q_{\sigma\triangle A},
\qquad
C_B^{(\tau)*}C_B^{(\tau)}
=
Q_{\tau\triangle B},
\]
and hence
\[
\mathbb P_{\mathfrak C}(A,B\mid\sigma,\tau)
=
\frac{1}{D}
\operatorname{Tr}\!\left(
Q_{\sigma\triangle A}
Q_{\tau\triangle B}
\right),
\]
which proves \eqref{eq:qbn-update-correlation}.

If
\(\sigma\triangle A=\tau\triangle B=\zeta\), then
\(Q_\zeta^2=Q_\zeta\), so
\[
\mathbb P_{\mathfrak C}(A,B\mid\sigma,\tau)
=
\frac{1}{D}\operatorname{Tr}(Q_\zeta)
=
\frac{\operatorname{rank}Q_\zeta}{D}.
\]
This proves \eqref{eq:meeting-branch-probability}.
\end{proof}

\begin{proposition}
\label{prop:resolution-sufficiency}
If a finite-dimensional projective \(t\)-sphere resolution exists,
then the associated update strategy satisfies
\begin{equation}
\omega_{d,t}(\mathbb P_{\mathfrak C})
=
\omega_{\mathrm{NS}}(d,t)
=
\frac{\binom{2t}{t}}{\binom dt}.
\label{eq:qbn-attains-ns}
\end{equation}
\end{proposition}

\begin{proof}
Write \(r_\zeta:=\operatorname{rank}Q_\zeta\). Taking traces in
\eqref{eq:projective-t-sphere} and summing over all
\(\sigma\in\Gamma_{d-1}\) gives
\begin{equation}
\sum_{\zeta\in\Gamma_{d-1}}r_\zeta
=
\frac{2^dD}{\binom dt}.
\label{eq:total-projective-rank}
\end{equation}

For a fixed meeting vertex \(\zeta\), successful branches are
determined by \(A,B\in\Omega_t\) with \(A\cap B=\varnothing\), through
\(\sigma=\zeta\triangle A\) and
\(\tau=\zeta\triangle B\). There are
\(\binom dt\binom{d-t}{t}\) such ordered pairs. Hence, by
\eqref{eq:meeting-branch-probability} and
\eqref{eq:total-projective-rank},
\begin{align*}
\omega_{d,t}(\mathbb P_{\mathfrak C})
&=
\frac{
\binom dt\binom{d-t}{t}
}{
2^d\binom d{2t}D
}
\sum_{\zeta\in\Gamma_{d-1}}r_\zeta
\\
&=
\frac{\binom{d-t}{t}}{\binom d{2t}}
=
\frac{\binom{2t}{t}}{\binom dt}.
\end{align*}
Lemma~\ref{lem:NS-general} identifies this value with
\(\omega_{\mathrm{NS}}(d,t)\).
\end{proof}

We now prove the converse.

\begin{theorem}
\label{thm:quantum-saturation-rigidity}
Assume \(d>2t\). If a finite-dimensional quantum strategy attains
\(\omega_{\mathrm{NS}}(d,t)\), then a finite-dimensional projective
\(t\)-sphere resolution exists.
\end{theorem}

\begin{proof}
By purification and local Naimark dilation, we may represent the
strategy by a pure bipartite state \(\ket\psi\) and projective
measurements
\[
\{E_{\sigma,A}:A\in\Omega_t\},
\qquad
\{F_{\tau,B}:B\in\Omega_t\}.
\]

Since the strategy attains the no-signalling bound, the equality
condition \eqref{eq:NS-saturation-equality} gives
\begin{equation}
\mathbb P(A,B\mid\sigma,\tau)
=
\mathbb P_{\mathrm A}(A\mid\sigma)
=
\mathbb P_{\mathrm B}(B\mid\tau)
\label{eq:winning-marginal-saturation}
\end{equation}
whenever \(A\cap B=\varnothing\) and
\(\tau=\sigma\triangle A\triangle B\). Consequently,
\[
\left\|
(E_{\sigma,A}\otimes I)\ket\psi
-
(I\otimes F_{\tau,B})\ket\psi
\right\|^2
=
0,
\]
and hence
\begin{equation}
(E_{\sigma,A}\otimes I)\ket\psi
=
(I\otimes F_{\tau,B})\ket\psi.
\label{eq:projector-synchronization}
\end{equation}

Let \(\mathcal S_{\mathrm A}\) and
\(\mathcal S_{\mathrm B}\) denote the Schmidt supports of
\(\ket\psi\). Since \(d>2t\), every \(A\in\Omega_t\) admits a disjoint
\(B\in\Omega_t\), so \eqref{eq:projector-synchronization} applies to
every measurement outcome. In Schmidt bases, let
\[
S:\mathcal S_{\mathrm B}\longrightarrow\mathcal S_{\mathrm A}
\]
be the invertible Schmidt map associated with \(\ket\psi\).
Equation~\eqref{eq:projector-synchronization} implies
\(E_{\sigma,A}S=SF_{\tau,B}^{T}\). Thus the measurement operators
preserve the Schmidt supports, and we may restrict to these supports
without changing the correlation. The state then has full Schmidt
rank, so
\begin{equation}
X\longmapsto(X\otimes I)\ket\psi
\quad\text{is injective}.
\label{eq:full-schmidt-injectivity}
\end{equation}

Fix \(\zeta\in\Gamma_{d-1}\). For disjoint
\(A,B\in\Omega_t\), set
\(\sigma=\zeta\triangle A\) and
\(\tau=\zeta\triangle B\). Then
\eqref{eq:projector-synchronization} becomes
\begin{equation}
(E_{\zeta\triangle A,A}\otimes I)\ket\psi
=
(I\otimes F_{\zeta\triangle B,B})\ket\psi.
\label{eq:meeting-point-synchronization}
\end{equation}

Let \(A,A'\in\Omega_t\). By
Proposition~\ref{prop:even-disjointness-chain}, there exist
\[
A=A_0,B_0,A_1,\ldots,B_{m-1},A_m=A'
\]
such that
\(A_i\cap B_i=\varnothing\) and
\(B_i\cap A_{i+1}=\varnothing\). Applying
\eqref{eq:meeting-point-synchronization} to the two disjoint pairs
\((A_i,B_i)\) and \((A_{i+1},B_i)\) gives
\[
(E_{\zeta\triangle A_i,A_i}\otimes I)\ket\psi
=
(E_{\zeta\triangle A_{i+1},A_{i+1}}\otimes I)\ket\psi.
\]
Iterating along the chain yields
\[
(E_{\zeta\triangle A,A}\otimes I)\ket\psi
=
(E_{\zeta\triangle A',A'}\otimes I)\ket\psi.
\]
By \eqref{eq:full-schmidt-injectivity},
\[
E_{\zeta\triangle A,A}
=
E_{\zeta\triangle A',A'}.
\]
Thus the projector depends only on the common final vertex \(\zeta\).
Define
\[
Q_\zeta
:=
E_{\zeta\triangle A,A},
\qquad
A\in\Omega_t.
\]

For fixed \(\sigma\in\Gamma_{d-1}\), taking
\(\zeta=\sigma\triangle A\) gives
\(E_{\sigma,A}=Q_{\sigma\triangle A}\). Measurement completeness
therefore implies
\[
\sum_{A\in\Omega_t}
Q_{\sigma\triangle A}
=
I,
\qquad
\sigma\in\Gamma_{d-1}.
\]
Hence
\(\{Q_\zeta:\zeta\in\Gamma_{d-1}\}\)
is a projective \(t\)-sphere resolution.
\end{proof}

\begin{remark}
Theorem~\ref{thm:quantum-saturation-rigidity} applies to arbitrary
finite-dimensional quantum strategies for the present rendezvous
problem; no QBN realization is assumed in the converse direction.
Thus every finite-dimensional strategy attaining the no-signalling
value must induce a projective \(t\)-sphere resolution. The result
concerns exact finite-dimensional attainment and does not address
infinite-dimensional strategies or non-attaining sequences approaching
\(\omega_{\mathrm{NS}}(d,t)\).
\end{remark}

The preceding results yield the following characterization.

\begin{corollary}
\label{cor:main-characterization}
Assume \(d>2t\). The following statements are equivalent:
\begin{enumerate}[(i)]
\item a finite-dimensional quantum strategy attains
\(\omega_{\mathrm{NS}}(d,t)\);

\item a finite-dimensional projective \(t\)-sphere resolution exists;

\item a finite-dimensional rendezvous-compatible \(t\)-coin field
exists.
\end{enumerate}
\end{corollary}
\section{Quantum optimality and quantum advantage}
\label{sec:quantum-optimality-advantage}

We now apply the saturation characterization to determine
when quantum strategies attain the no-signalling value and when they
yield a strict advantage over classical strategies. We first consider
\(t=1\), which is the nearest-neighbor rendezvous problem.

In the QBN realization, the update
\(
\sigma\longmapsto\sigma\triangle\{j\}
\)
is implemented by \(\Xi_j\), and the projective sphere condition
reduces to
\begin{equation}
\sum_{j=0}^{d-1}
Q_{\sigma\triangle\{j\}}
=
I,
\qquad
\sigma\in\Gamma_{d-1}.
\label{eq:one-step-projective-resolution}
\end{equation}
By Corollary~\ref{cor:main-characterization}, the existence of such a
projective resolution is equivalent to finite-dimensional quantum
attainment of the no-signalling value
\[
\omega_{\mathrm{NS}}(d)=\frac{2}{d}.
\]

For $t=1$, Proposition~\ref{prop:classical-general} identifies classical
saturation with the existence of a scalar $1$-sphere resolution.
Writing
\[
S:=\{\zeta\in\Gamma_{d-1}:q(\zeta)=1\},
\]
the scalar resolution condition is equivalent to requiring that every
vertex of the hypercube have exactly one neighbor in $S$. Thus scalar
$1$-sphere resolutions are precisely perfect total dominating sets of
the hypercube. By the known characterization of such sets
\cite{KwonLee2014}, they exist if and only if $d$ is a power of two.
We therefore obtain the following corollary.

\begin{corollary}
\label{cor:classical-one-step}
For $d\ge2$,
\begin{equation}
\omega_{\mathrm C}(d)
=
\frac{2}{d}
\quad\Longleftrightarrow\quad
d=2^m
\quad\text{for some }m\ge1.
\label{eq:classical-one-step-saturation}
\end{equation}
\end{corollary}

In the quantum setting, the projections need not be scalar. A real
Hadamard matrix provides an explicit rank-one solution of
\eqref{eq:one-step-projective-resolution}, and hence an explicit
coin field.

Let
\[
H=(h_{jk})_
{0\le j\le d-1,\ 1\le k\le d}
\]
be a real Hadamard matrix of order \(d\), so that
\[
h_{jk}\in\{\pm1\},
\qquad
HH^T=dI_d.
\]
On \(\mathcal K=\mathbb C^d\), with standard basis
\(\{e_k\}_{k=1}^d\), define
\begin{equation}
R_je_k
:=
h_{jk}e_k,
\qquad
u
:=
\frac1{\sqrt d}
\sum_{k=1}^d e_k.
\label{eq:hadamard-reflections}
\end{equation}
The operators \(R_j\) are commuting self-adjoint unitaries. For
\(\sigma\in\Gamma_{d-1}\), set
\begin{equation}
R_\sigma
:=
\prod_{j\in\sigma}R_j,
\qquad
Q_\sigma
:=
|R_\sigma u\rangle
\langle R_\sigma u|.
\label{eq:hadamard-projections}
\end{equation}

\begin{proposition}
\label{prop:hadamard-coin-field}
The projections
\[
\{Q_\sigma:\sigma\in\Gamma_{d-1}\}
\]
satisfy \eqref{eq:one-step-projective-resolution}. Thus
\begin{equation}
C_j^{(\sigma)}
:=
Q_{\sigma\triangle\{j\}},
\qquad
j\in N_{d-1},
\label{eq:hadamard-coin-system}
\end{equation}
defines a rendezvous-compatible one-step coin field. The
corresponding local update is
\begin{equation}
W_{\mathcal C^{(\sigma)}}^{[1]}
=
\sum_{j=0}^{d-1}
\Xi_j\otimes
Q_{\sigma\triangle\{j\}}.
\label{eq:hadamard-qbn-walk}
\end{equation}
\end{proposition}

\begin{proof}
Fix \(\sigma\in\Gamma_{d-1}\). For distinct
\(j,\ell\in N_{d-1}\),
\begin{align*}
\left\langle
R_{\sigma\triangle\{j\}}u,
R_{\sigma\triangle\{\ell\}}u
\right\rangle
&=
\langle R_ju,R_\ell u\rangle\\
&=
\frac1d
\sum_{k=1}^d
h_{jk}h_{\ell k}\\
&=
0.
\end{align*}
Each vector has norm one. Hence
\[
\{R_{\sigma\triangle\{j\}}u:
j\in N_{d-1}\}
\]
is an orthonormal basis of \(\mathbb C^d\), and therefore
\[
\sum_{j=0}^{d-1}
Q_{\sigma\triangle\{j\}}
=
I.
\]
This proves \eqref{eq:one-step-projective-resolution}.

By Theorem~\ref{thm:coin-resolution-equivalence},
\eqref{eq:hadamard-coin-system} defines a rendezvous-compatible
one-step coin field. Proposition~\ref{prop:t-coordinate-qbn-operator}
then gives the update
\eqref{eq:hadamard-qbn-walk}.
\end{proof}

For the initial Bernoulli basis state \(Z_\sigma\), the operator
\eqref{eq:hadamard-qbn-walk} acts as
\[
W_{\mathcal C^{(\sigma)}}^{[1]}
(Z_\sigma\otimes\xi)
=
\sum_{j=0}^{d-1}
Z_{\sigma\triangle\{j\}}
\otimes
Q_{\sigma\triangle\{j\}}\xi,
\qquad
\xi\in\mathcal K.
\]
Thus the Hadamard projections determine simultaneously the possible
nearest-neighbor QBN updates and their coin-space components.

\begin{corollary}[One-step quantum advantage]
\label{cor:one-step-advantage}
If a real Hadamard matrix of order \(d\) exists, then
\begin{equation}
\omega_{\mathrm Q}(d)
=
\omega_{\mathrm{NS}}(d)
=
\frac2d.
\label{eq:Hadamard-quantum-value}
\end{equation}
If, in addition, \(d\) is not a power of two, then
\begin{equation}
\omega_{\mathrm C}(d)
<
\omega_{\mathrm Q}(d)
=
\frac2d.
\label{eq:one-step-strict-advantage}
\end{equation}
\end{corollary}

\begin{proof}
Proposition~\ref{prop:hadamard-coin-field} gives a finite-dimensional
projective \(1\)- sphere resolution. Hence
Proposition~\ref{prop:resolution-sufficiency} yields
\[
\omega_{\mathrm Q}(d)
=
\omega_{\mathrm{NS}}(d)
=
\frac2d.
\]
If \(d\) is not a power of two,
Corollary~\ref{cor:classical-one-step} gives
\[
\omega_{\mathrm C}(d)<\frac2d.
\]
\end{proof}

We now turn to \(t\ge2\). In contrast with the \(t=1\) regime, the
projective sphere condition is obstructed in finite dimension when
\(d>2t\). Through
Theorem~\ref{thm:coin-resolution-equivalence}, this is equivalently
an obstruction to finite-dimensional rendezvous-compatible 
\(t\)-coin fields.

\begin{theorem}
\label{thm:higher-radius-obstruction}
Let \(t\ge2\) and \(d>2t\). Then no finite-dimensional projective
\(t\)-sphere resolution exists.
\end{theorem}

\begin{proof}
Suppose that
\[
\{Q_\zeta:\zeta\in\Gamma_{d-1}\}
\]
is a finite-dimensional projective \(t\)-sphere resolution on
\(\mathcal K\).

For \(\sigma\in\Gamma_{d-1}\), recall the QBN translation
\[
\Xi_\sigma
:=
\prod_{j\in\sigma}\Xi_j,
\qquad
\Xi_\sigma Z_\zeta
=
Z_{\zeta\triangle\sigma}.
\]
Using these translations, symmetrize the projective family by setting
\[
\widehat{\mathcal K}
:=
\bigoplus_{\sigma\in\Gamma_{d-1}}\mathcal K,
\qquad
\widehat Q_\zeta
:=
\bigoplus_{\sigma\in\Gamma_{d-1}}
Q_{\zeta\triangle\sigma}.
\]
Then
\(\{\widehat Q_\zeta:\zeta\in\Gamma_{d-1}\}\)
is again a projective \(t\)-sphere resolution. Indeed, for every
\(\tau\in\Gamma_{d-1}\),
\[
\sum_{A\in\Omega_t}
\widehat Q_{\tau\triangle A}
=
\bigoplus_{\sigma\in\Gamma_{d-1}}
\sum_{A\in\Omega_t}
Q_{\tau\triangle\sigma\triangle A}
=
I_{\widehat{\mathcal K}}.
\]

Since \(\zeta\triangle\sigma\) runs through
\(\Gamma_{d-1}\) as \(\sigma\) varies, the projections
\(\widehat Q_\zeta\) all have the same rank, say \(r>0\). Writing
\(\widehat D:=\dim\widehat{\mathcal K}\), the orthogonality of the
projections on each \(t\)-sphere gives
\begin{equation}
\widehat D
=
\binom dt r.
\label{eq:symmetrized-rank-identity}
\end{equation}

Now consider
\[
\mathcal J_t:=\Omega_t\cup\Omega_{t-2},
\]
corresponding to the Bernoulli basis vectors in the layers
\(V_t\) and \(V_{t-2}\). We claim that for any distinct
\(\zeta,\zeta'\in\mathcal J_t\), there exists
\(\tau\in\Gamma_{d-1}\) such that
\[
|\tau\triangle\zeta|
=
|\tau\triangle\zeta'|
=
t.
\]

Indeed, since \(\zeta\) and \(\zeta'\) have the same parity, write
\[
|\zeta\triangle\zeta'|=2s,
\qquad
1\le s\le t.
\]
Choose a set \(T\in\Omega_t\) containing exactly \(s\) elements of
\(\zeta\triangle\zeta'\) and \(t-s\) elements outside it. Such a
choice is possible because
\[
d-2s>2t-2s\ge t-s.
\]
Set
\[
\tau:=\zeta\triangle T.
\]
Then
\[
|\tau\triangle\zeta|=|T|=t,
\]
while
\[
|\tau\triangle\zeta'|
=
|T\triangle(\zeta\triangle\zeta')|
=
t.
\]
Thus \(\widehat Q_\zeta\) and
\(\widehat Q_{\zeta'}\) occur in the same projective \(t\)-sphere
resolution and are therefore orthogonal.

Hence
\[
\{\widehat Q_\zeta:\zeta\in\mathcal J_t\}
\]
is a pairwise orthogonal family of rank-\(r\) projections. It follows
that
\[
\widehat D
\ge
|\mathcal J_t|r
=
\left(
\binom dt+\binom d{t-2}
\right)r
>
\binom dt r,
\]
contradicting \eqref{eq:symmetrized-rank-identity}. Therefore no
finite-dimensional projective \(t\)-sphere resolution exists.
\end{proof}
\section{Conclusion}
\label{sec:conclusion}

This paper formulates one-step rendezvous on the hypercube as a two-agent QBN quantum walk with step length \(t\), and identifies projective \(t\)-sphere resolutions as the central structure governing exact rendezvous. The QBN shifts realize
fixed-cardinality Hamming updates, while reindexing the coin projections
by the common final vertex converts the rendezvous problem into a
geometric resolution problem on Hamming spheres. For \(d>2t\), scalar
\(t\)-sphere resolutions characterize classical saturation, while
projective \(t\)-sphere resolutions characterize finite-dimensional
quantum attainment of the no-signalling value. In the nearest-neighbor
case \(t=1\), classical saturation occurs exactly for \(d=2^m\), whereas
real Hadamard matrices yield quantum saturation; hence Hadamard orders
that are not powers of two give a strict quantum advantage. For
\(t\ge2\) and \(d>2t\), a rank obstruction rules out finite-dimensional
projective \(t\)-sphere resolutions.

The main insight is that exact quantum optimality is governed by the
geometry of the possible meeting vertices rather than by the local
update rule alone. The distinction between scalar and projective sphere
resolutions explains how quantum strategies can attain the
no-signalling optimum when classical strategies cannot, while the
higher-cardinality obstruction shows that this mechanism itself breaks
down once \(t\ge2\). It remains open whether
\(\omega_{\mathrm Q}(d,t)\) is then strictly below
\(\omega_{\mathrm{NS}}(d,t)\), and whether a strict quantum advantage
over classical strategies persists in this regime. These questions make
projective Hamming-sphere resolutions a natural structural framework
for studying the limits of quantum coordination on high-dimensional
hypercubes.
\section*{Author declarations}

\subsection*{Conflict of interest}

The author has no conflicts to disclose.

\subsection*{Data Availability}

Data sharing is not applicable to this article as no new data were created or analyzed in this study.
\appendix

\section{Proof of the no-signalling benchmark}
\label{app:no-signalling}

\begin{proof}[Proof of Lemma~\ref{lem:NS-general}]
Let \(\mathbb P(A,B\mid\sigma,\tau)\) be an arbitrary
no-signalling correlation.

For fixed \(\sigma\in\Gamma_{d-1}\) and \(A\in\Omega_t\), every
successful configuration is obtained by choosing
\(B\in\Omega_t\) with \(A\cap B=\varnothing\) and setting
\(\tau=\sigma\triangle A\triangle B\). There are
\(\binom{d-t}{t}\) such choices of \(B\).

By no-signalling,
\(\mathbb P(A,B\mid\sigma,\tau)
\le \mathbb P_{\mathrm A}(A\mid\sigma)\).
Hence
\[
\sum_{\sigma\in\Gamma_{d-1}}
\sum_{A\in\Omega_t}
\sum_{\substack{B\in\Omega_t\\A\cap B=\varnothing}}
\mathbb P
\left(
A,B
\mid
\sigma,
\sigma\triangle A\triangle B
\right)
\le
2^d\binom{d-t}{t}.
\]
Since the number of  ordered input pairs is
\(2^d\binom d{2t}\), we obtain
\begin{equation}
\omega_{\mathrm{NS}}(d,t)
\le
\frac{\binom{d-t}{t}}{\binom d{2t}}.
\label{eq:NS-upper-bound}
\end{equation}

We now construct a no-signalling correlation attaining this bound.
Set \(n:=\binom dt\). For each \(\eta\in\Omega_{2t}\), the map
\(A\mapsto\eta\setminus A\) is a bijection on
\(\{A\in\Omega_t:A\subseteq\eta\}\). Extend it arbitrarily to a
permutation \(\pi_\eta:\Omega_t\to\Omega_t\).

For an  input pair \((\sigma,\tau)\), write
\(\eta:=\sigma\triangle\tau\) and define
\[
\mathbb P(A,B\mid\sigma,\tau)
:=
\begin{cases}
\dfrac1n,
&
B=\pi_\eta(A),
\\[2mm]
0,
&
\text{otherwise}.
\end{cases}
\]
Since \(\pi_\eta\) is a permutation, both marginals are uniform:
\(\mathbb P_{\mathrm A}(A\mid\sigma)
=\mathbb P_{\mathrm B}(B\mid\tau)=1/n\). Thus the correlation is
no-signalling.

For fixed \(\eta\), success occurs precisely when
\(A\subseteq\eta\), in which case \(B=\eta\setminus A\). There are
\(\binom{2t}{t}\) such choices of \(A\). Therefore
\[
\omega_{\mathrm{NS}}(d,t)
\ge
\frac{\binom{2t}{t}}{\binom dt}.
\]
Using
\[
\frac{\binom{d-t}{t}}{\binom d{2t}}
=
\frac{\binom{2t}{t}}{\binom dt},
\]
together with \eqref{eq:NS-upper-bound}, gives
\[
\omega_{\mathrm{NS}}(d,t)
=
\frac{\binom{d-t}{t}}{\binom d{2t}}
=
\frac{\binom{2t}{t}}{\binom dt}.
\]

We also record the equality condition needed later. If a
no-signalling correlation attains this value, then equality must hold
termwise in the preceding upper bound. Hence, whenever
\(A\cap B=\varnothing\) and
\(\tau=\sigma\triangle A\triangle B\),
\begin{equation}
\mathbb P(A,B\mid\sigma,\tau)
=
\mathbb P_{\mathrm A}(A\mid\sigma)
=
\mathbb P_{\mathrm B}(B\mid\tau).
\label{eq:NS-saturation-equality}
\end{equation}
The second equality follows by interchanging the roles of the two
players.
\end{proof}


\section{Proof of the classical saturation proposition}
\label{app:classical-saturation}

We first establish the local estimate used in the proof. On the QBN
fixed-cardinality layer
\(V_t=\operatorname{span}\{Z_A:A\in\Omega_t\}\), define
\[
\mathsf D_tZ_A
:=
\sum_{\substack{B\in\Omega_t\\A\cap B=\varnothing}} Z_B.
\]
Set \(r_t:=\binom{d-t}{t}\) and
\(\mathbf 1:=\sum_{A\in\Omega_t}Z_A\).

\begin{lemma}
\label{lem:disjointness-operator}
Assume \(d>2t\). Then \(\|\mathsf D_t\|=r_t\), and the singular
subspace associated with the maximal singular value \(r_t\) is
\(\operatorname{span}\{\mathbf 1\}\). Consequently, for nonempty
\(X,Y\subseteq\Omega_t\),
\begin{equation}
\left\langle
\mathbf 1_X,
\mathsf D_t\mathbf 1_Y
\right\rangle
\le
r_t\sqrt{|X||Y|},
\label{eq:disjointness-bound}
\end{equation}
where \(\mathbf 1_X:=\sum_{A\in X}Z_A\). Equality holds if and only if
\(X=Y=\Omega_t\).
\end{lemma}

\begin{proof}
Each row and column of the matrix of \(\mathsf D_t\) in the Bernoulli
basis contains exactly \(r_t\) entries equal to \(1\). Hence
\[
\|\mathsf D_t\|
\le
\sqrt{\|\mathsf D_t\|_1\|\mathsf D_t\|_\infty}
=
r_t.
\]
Since \(\mathsf D_t\mathbf 1=r_t\mathbf 1\), it follows that
\(\|\mathsf D_t\|=r_t\).

Let \(f=\sum_{A\in\Omega_t}f_AZ_A\) satisfy
\(\mathsf D_tf=\lambda f\) with \(|\lambda|=r_t\). Choose
\(A_0\in\Omega_t\) such that
\(|f_{A_0}|=\max_{A\in\Omega_t}|f_A|\). Then
\[
r_t|f_{A_0}|
=
\left|
\sum_{\substack{B\in\Omega_t\\A_0\cap B=\varnothing}}
f_B
\right|
\le
\sum_{\substack{B\in\Omega_t\\A_0\cap B=\varnothing}}
|f_B|
\le
r_t|f_{A_0}|.
\]
Thus equality holds throughout. Repeating the argument along the even
disjointness chains of Proposition~\ref{prop:even-disjointness-chain}
gives \(|f_A|=|f_{A_0}|\) for all \(A\in\Omega_t\).

If \(\lambda=r_t\), equality in the triangle inequality gives
\(f_A=f_B\) for every disjoint pair \(A,B\). The same chain argument
shows that all coefficients are equal, so
\(f\in\operatorname{span}\{\mathbf 1\}\).

If \(\lambda=-r_t\), then \(f_B=-f_A\) for every disjoint pair
\(A,B\). Choose such a pair. The direct disjointness transition gives
\(f_B=-f_A\), while an even disjointness chain between the same two
indices, supplied by Proposition~\ref{prop:even-disjointness-chain},
gives \(f_B=f_A\), a contradiction. Thus \(-r_t\) is not an
eigenvalue.

It follows that the maximal singular subspace is
\(\operatorname{span}\{\mathbf 1\}\).

Now
\[
\left\langle
\mathbf 1_X,
\mathsf D_t\mathbf 1_Y
\right\rangle
\le
\|\mathbf 1_X\|
\,
\|\mathsf D_t\mathbf 1_Y\|
\le
r_t
\|\mathbf 1_X\|
\,
\|\mathbf 1_Y\|.
\]
By the orthonormality of the Bernoulli basis,
\(\|\mathbf 1_X\|^2=|X|\) and
\(\|\mathbf 1_Y\|^2=|Y|\), which proves
\eqref{eq:disjointness-bound}.

If equality holds, then \(\mathbf 1_Y\) belongs to the maximal
singular subspace, so \(Y=\Omega_t\). Equality in the first
Cauchy--Schwarz inequality then gives \(X=\Omega_t\). The converse is
immediate.
\end{proof}

\begin{proof}[Proof of Proposition~\ref{prop:classical-general}]
Since the classical correlation set is the convex hull of deterministic
correlations and the success functional is linear, it suffices to
consider deterministic response functions
\(A(\sigma),B(\tau)\in\Omega_t\). Define
\(F(\sigma):=\sigma\triangle A(\sigma)\) and
\(G(\tau):=\tau\triangle B(\tau)\).

For each \(\zeta\in\Gamma_{d-1}\), set
\[
X_\zeta
:=
\{A\in\Omega_t:F(\zeta\triangle A)=\zeta\},
\qquad
Y_\zeta
:=
\{B\in\Omega_t:G(\zeta\triangle B)=\zeta\}.
\]
The correspondence between inputs and their final vertices gives
\begin{equation}
\sum_{\zeta\in\Gamma_{d-1}}|X_\zeta|
=
\sum_{\zeta\in\Gamma_{d-1}}|Y_\zeta|
=
2^d.
\label{eq:classical-fiber-sums}
\end{equation}

Let \(W_\zeta\) denote the number of winning ordered input pairs whose
common final vertex is \(\zeta\). Such a pair corresponds precisely to
\(A\in X_\zeta\), \(B\in Y_\zeta\), and
\(A\cap B=\varnothing\). Indeed, the corresponding inputs are
\(\sigma=\zeta\triangle A\) and \(\tau=\zeta\triangle B\), and since
\(A,B\in\Omega_t\),
\[
|\sigma\triangle\tau|=2t
\quad\Longleftrightarrow\quad
A\cap B=\varnothing.
\]
Therefore
\[
W_\zeta
=
\left\langle
\mathbf 1_{X_\zeta},
\mathsf D_t\mathbf 1_{Y_\zeta}
\right\rangle.
\]
By Lemma~\ref{lem:disjointness-operator},
\(W_\zeta\le r_t\sqrt{|X_\zeta||Y_\zeta|}\). Summing over
\(\zeta\), applying Cauchy--Schwarz, and using
\eqref{eq:classical-fiber-sums}, we obtain
\[
\sum_{\zeta\in\Gamma_{d-1}}W_\zeta
\le
2^dr_t.
\]
Since there are \(2^d\binom d{2t}\)  ordered input pairs,
\[
\omega_{\mathrm C}(d,t)
\le
\frac{\binom{d-t}{t}}{\binom d{2t}}
=
\omega_{\mathrm{NS}}(d,t).
\]

Suppose first that
\(\omega_{\mathrm C}(d,t)=\omega_{\mathrm{NS}}(d,t)\). Then equality
holds in both inequalities above. Equality in the global
Cauchy--Schwarz inequality, together with
\eqref{eq:classical-fiber-sums}, gives
\(|X_\zeta|=|Y_\zeta|\) for every \(\zeta\). For every \(\zeta\)
with these sets nonempty, equality must also hold in
\eqref{eq:disjointness-bound}. Hence
\(X_\zeta=Y_\zeta=\Omega_t\).

Call such a vertex \(\zeta\) active. For every active \(\zeta\),
\(F(\zeta\triangle A)=\zeta\) for all \(A\in\Omega_t\). Thus the
\(t\)-sphere
\[
\mathsf S_t(\zeta)
:=
\{\zeta\triangle A:A\in\Omega_t\}
\]
is mapped entirely to \(\zeta\). Since \(F\) is single-valued and
every input has a final vertex, the \(t\)-spheres centered at the
active vertices form a partition of \(\Gamma_{d-1}\).

Define
\[
q(\zeta)
:=
\begin{cases}
1,
&
\zeta\text{ is active},
\\[1mm]
0,
&
\text{otherwise}.
\end{cases}
\]
The partition property is equivalent to
\[
\sum_{A\in\Omega_t}
q(\sigma\triangle A)
=
1,
\qquad
\sigma\in\Gamma_{d-1}.
\]
Hence \(q\) is a scalar \(t\)-sphere resolution.

Conversely, suppose that
\(q:\Gamma_{d-1}\to\{0,1\}\) satisfies
\[
\sum_{A\in\Omega_t}q(\sigma\triangle A)=1,
\qquad
\sigma\in\Gamma_{d-1}.
\]
Then every input \(\sigma\) belongs to a unique \(t\)-sphere centered
at a vertex \(\zeta(\sigma)\) with \(q(\zeta(\sigma))=1\). The two
players agree in advance on the deterministic response rule
\[
A(\sigma)
:=
\sigma\triangle\zeta(\sigma)
\in\Omega_t.
\]

Fix \(\sigma\), and write \(A:=A(\sigma)\) and
\(\zeta:=\zeta(\sigma)\). For every
\(B\in\Omega_t\) with \(A\cap B=\varnothing\), define
\(\tau:=\sigma\triangle A\triangle B\). Then
\(|\sigma\triangle\tau|=2t\) and
\(\tau\triangle\zeta=B\). Thus \(\tau\) lies on the same active
\(t\)-sphere, so \(\zeta(\tau)=\zeta\) and \(A(\tau)=B\).
Consequently,
\[
\sigma\triangle A(\sigma)
=
\zeta
=
\tau\triangle A(\tau),
\]
and the players meet.

For fixed \(\sigma\), there are \(\binom{d-t}{t}\) choices of
\(B\) disjoint from \(A\), among \(\binom d{2t}\)  relative
sets. Hence
\[
\omega_{\mathrm C}(d,t)
=
\frac{\binom{d-t}{t}}{\binom d{2t}}
=
\omega_{\mathrm{NS}}(d,t).
\]
This proves the proposition.
\end{proof}

\end{document}